\documentclass{article}
\usepackage[margin=1.25in]{geometry}
\usepackage{graphicx}
\usepackage{amsmath}
\usepackage{amssymb}
\usepackage{mathtools}
\usepackage{parskip}
\usepackage[hidelinks]{hyperref}
\usepackage{natbib}
\usepackage{subcaption}
\usepackage{wrapfig}
\usepackage{color}
\usepackage{amsthm}
\newtheorem{theorem}{Theorem}
\newtheorem{corollary}{Corollary}
\emergencystretch=\maxdimen
\begin{document}
\title{
	On maximizing VAM for a given power \\
	{\Large Slope, cadence, force and gear-ratio considerations}
}
\author{
	Len Bos%
	\footnote{
		Universit\'a di Verona, Italy: {\tt leonardpeter.bos@univr.it}
	}\,,
	Michael A. Slawinski%
	\footnote{
		Memorial University of Newfoundland, Canada: {\tt mslawins@mac.com}; {\tt theodore.stanoev@gmail.com}
	}\,,
	\addtocounter{footnote}{-1}
	Theodore Stanoev\footnotemark
}
\date{}
\maketitle
\begin{abstract}
The {\it velocit\`a ascensionale media} (VAM) is measurement that quantifies a cyclist's climbing ability.
It depends on both the ground speed of a bicycle-cyclist system and the slope of an incline.
To maximize the ascent speed, the solution to the brachistochrone problem determines that the optimal curve for the incline is a straight line, which is a hill with a constant slope.
The maximum obtainable VAM value by a cyclist increases monotonically with the slope of the incline, but is limited by the maximum sustainable power.
These properties\,---\,which are theorems stemming from a standard mathematical model to account for the power required to propel a bicycle\,---\,constitute a mathematical-physics background upon which various strategies for the VAM maximization can be examined in the context of the maximum sustainable power as a function of both the gear ratio and cadence.
Recently established records provide an empirical support for these analytical results, which are based on theoretical considerations.
\end{abstract}
\section{Introduction}
\citet{Ferrari2003} introduced {\it velocit\`a ascensionale media} (VAM), mean ascent velocity, which is a measurement of the rate of ascent and quantifies a cyclist's climbing ability.

Commonly, VAM is built-in into a GPS bike computer~\citep[e.g.,][]{PolarM460}, as an instantaneous information, and is available {\it a posteriori} for many Strava segments~\citep{Strava2018}.
It is used in training manuals~\citep[e.g.,][]{Rotondi2012} and touring books~\citep[e.g.,][]{Barrett2016}.
An application to racing strategies is described by~\citet{Gifford2011}, for Stage~16 of the Tour de France in 2000.
Ferrari advised the USPS team not to chase Marco Pantani, in spite of a large gap, based on VAM measurements, which he considered unsustainable.
Ferrari's prediction was correct: Pantani was overtaken by the peloton.
Recently, due to the popularity of Everesting, VAM has become a point of conversation among cyclists eager to attempt this challenge.

In spite of its presence within the cycling community, VAM is rarely considered as a  research topic.
At the time of this article, an advanced search on Google Scholar for articles containing the exact phrase ``mean ascent velocity'' yields fewer than one hundred results, while restricting the search to include ``cycling'' yields fewer than ten. 
Therein, the sole consideration of VAM is by~\citet{CintiaEtAl2013}, for data mining using Strava.

VAM is a value devoid of information concerning the conditions under which it is obtained.
For instance, the same VAM value can be obtained by slowly climbing a steep hill or by riding quickly a moderate incline.
In this article, we use mathematical physics to provide a novel approach to enhance our understanding of VAM and to examine conditions for its maximization.

We begin this article by presenting a model that accounts for the power to maintain a given ground speed. 
Then, we prove that the maximization of VAM requires a constant slope, as a consequence of solving the brachistochrone problem.
We proceed by using the model to express VAM in terms of ground speed and slope, and prove, for a fixed power output, that VAM increases monotonically with the slope.
To gain an insight into these results, we discuss a numerical example, wherein we vary the values of cadence and gear ratio.
We conclude by suggesting possible consequences of our results for VAM-maximization strategies.
\section{Method}
In this article, we use methods rooted in mathematical physics and logical proof to determine the conditions required to maximize VAM, which we demonstrate in two steps.

The first step is to obtain the curve that minimizes the time between two points along an incline.
In other words, we consider the brachistochrone problem to determine the optimal hill conditions that minimize the ascent time.
The solution to the problem, which is steeped in the history of physics, laid the groundwork for the mid-eighteenth-century development of the calculus of variations by Leonhard Euler and Joseph-Louis Lagrange~\citep[e.g.,][pp.~19--20]{Coopersmith2017}.
In our solution, we avail of these developments and use the Euler-Lagrange equation to determine that the brachistochrone for the ascent is a straight line, which, given the nature of typical solutions to brachistochrone problems, is an unexpected result.

The second step requires a mathematical formulation to model a cyclist climbing a straight-line incline.
To that end, we consider a model that accounts for the power required to overcome the forces, $F_{\!\leftarrow}$\,, that oppose the translational motion of a bicycle-cyclist system along a hill of constant slope with ground speed~$V$~\citep[e.g.,][]{DanekEtAl2021}.
The model is
\begin{align}
	\label{eq:PV}
	P
	&=
	F_{\!\leftarrow}\,V
	\\
	\nonumber 
	&=
	\quad\frac{
		\overbrace{\vphantom{\left(V\right)^2}\,m\,g\sin\theta\,}^\text{change in elevation}
		+
		\!\!\!\overbrace{\vphantom{\left(V\right)^2}\quad m\,a\quad}^\text{change in speed}
		+
		\overbrace{\vphantom{\left(V\right)^2}
			{\rm C_{rr}}\!\!\!\underbrace{\,m\,g\cos\theta}_\text{normal force}
		}^\text{rolling resistance}
		+
		\overbrace{\vphantom{\left(V\right)^2}
			\,\tfrac{1}{2}\,{\rm C_{d}A}\,\rho\,
			(
				\!\!\underbrace{
					V+w_\leftarrow
				}_\text{air-flow speed}\!\!
			)^{2}\,
		}^\text{air resistance}
	}{
		\underbrace{\quad1-\lambda\quad}_\text{drivetrain efficiency}
	}\,V\,,
\end{align}
where $m$~is the mass of the bicycle-cyclist system, $g$~is the acceleration due to gravity, $\theta$ is the slope of the incline, $a$~is the change of speed, $\rm C_{rr}$~is the rolling-resistance coefficient, $\rm C_{d}A$~is the air-resistance coefficient, $\rho$~is the air density, $w_{\leftarrow}$~is the wind component opposing the motion and $\lambda$ is the drivetrain-resistance coefficient.
Considering an idealized setup\,---\,that is, a steady ride,~$a=0\,{\rm m/s^2}$\,, in windless conditions,~$w=0\,{\rm m/s}$\,---\,we write model~\eqref{eq:PV} as
\begin{equation}
	\label{eq:P}
	P=\dfrac{mg\left({\rm C_{rr}}\cos\theta+\sin\theta\right)+\tfrac{1}{2}\,{\rm C_{d}A}\,\rho\,V^2}{1-\lambda}V\,.
\end{equation}

Finally, to present the theorems associated with the first and second steps, we define 
\begin{equation}
\label{eq:VAM}
{\rm VAM} := 3600\,V(\theta)\,\sin\theta\,,
\end{equation}
which is stated in metres per hour.
Thus, for a given power, $P = P_0$\,, and using expression~\eqref{eq:P}, we write
\begin{equation}
\label{eq:PowerConstraint}
V^{3}+A(\theta)\,V + B = 0\,,
\end{equation}
where
\begin{equation}
	\label{eq:a}
	A(\theta) = \frac{mg\left({\rm C_{rr}}\cos\theta+\sin\theta\right)}{\tfrac{1}{2}{\rm C_{d}A}\,\rho}
\end{equation}
and
\begin{equation}
	\label{eq:b}
	B = -\frac{\left(1-\lambda\right)P_{0}}{\tfrac{1}{2}{\rm C_{d}A}\,\rho}\,.
\end{equation}


\section{Results}
\subsection{Brachistochrone problem}
\label{sec:brachistochrone}
Let us consider the brachistochrone problem, which, herein, consists of finding the curve along which the ascent between two points\,---\,under the assumption of fixed power\,---\,takes the least amount of time.
To solve the problem, we prove the following theorem.\\
\begin{theorem}
	\label{thm:Brachistochrone}
	The path of quickest ascent from $(0,0)$ to $(R,H)$\,, which is an elevation gain of~$H$ over a horizontal distance of~$R$\,, is the straight line, $y = (H/R)\,x$\,, where $0\leqslant x\leqslant R$ and $y(R) = H$\,.
\end{theorem}
\begin{proof}
	The relation between speed, slope and power is stated by expressions~\eqref{eq:PowerConstraint},~\eqref{eq:a} and~\eqref{eq:b}, which determines the ground speed $V=V(\theta)$\,, where $\theta$ is the slope, namely,
	\begin{equation*}
		\tan\theta = \frac{{\rm d}y}{{\rm d}x}\,,\quad
		\theta = \arctan\left(\frac{{\rm d}y}{{\rm d}x}\right)\,.
	\end{equation*}
	The traveltime over an infinitesimal distance, ${\rm d}s$\,, is ${\rm d}t = {\rm d}s/V$\,.
	Since ${\rm d}s^2={\rm d}x^2+{\rm d}y^2$\,, we write
	\begin{equation*}
		{\rm d}t = \frac{\sqrt{1+\left(\dfrac{{\rm d}y}{{\rm d}x}\right)^{2}}\,{\rm d}x}{V(\theta)}
		= \frac{\sqrt{1+\left(y'\right)^{2}}}{V(\arctan(y'))}\,{\rm d}x\,,
	\end{equation*}
	where $y':={\rm d}y/{\rm d}x$\,.
	Hence, the ascent time is
	\begin{equation*}
		T = \int\limits_0^R\frac{\sqrt{1+\left(y'\right)^{2}}}{V(\arctan(y'))}\,{\rm d}x
		=:\int\limits_0^R\underbrace{\frac{\sqrt{1+p^{2}}}{V(\arctan(p))}}_F\,{\rm d}x\,,
	\end{equation*}
	where $p:=y'(x)$\,.
	Thus, we write the integrand as a function of three independent variables,
	\begin{equation*}
		T=: \int\limits_0^RF(x,y,p)\,{\rm d}x\,.
	\end{equation*}
	To minimize  the traveltime, we invoke the Euler-Lagrange equation,
	\begin{equation*}
		\frac{\partial}{\partial y}F(x,y,p) = \frac{\rm d}{{\rm d}x}\frac{\partial}{\partial p}F(x,y,p)\,.
	\end{equation*}
	Since $F$ is not an explicit function of $y$\,, the left-hand side is zero,
	\begin{equation}
	\label{eq:ConsQuant}
		0 = \frac{\rm d}{{\rm d}x}\left\{\frac{\partial}{\partial p}\frac{\sqrt{1+p^{2}}}{V(\arctan(p))}\right\}\,,
	\end{equation}
	which means that the term in braces is a constant.
	To proceed, we use the fact that
	\begin{equation*}
		p = \tan\theta \implies \frac{\partial}{\partial p} = \frac{{\rm d}\theta}{{\rm d}p}\frac{\partial}{\partial\theta} = \cos^{2}\theta\,\frac{\partial}{\partial\theta}
	\end{equation*}
	and
	\begin{align*}
		\frac{\partial}{\partial p}\frac{\sqrt{1+p^{2}}}{V(\arctan(p))}
		&= \cos^{2}\theta\,\frac{\partial}{\partial\theta}\frac{\sqrt{1+p^{2}}}{V(\arctan(p))}
		= \cos^{2}\theta\,\frac{\partial}{\partial\theta}\frac{\sqrt{1+\tan^{2}\theta}}{V(\theta)}
		= \cos^{2}\theta\,\frac{\partial}{\partial\theta}\frac{\sqrt{\sec\theta}}{V(\theta)} \\
		&= \cos^{2}\theta\,\frac{\rm d}{{\rm d}\theta}\frac{1}{\cos\theta\,V(\theta)}\,,
	\end{align*}
	for $0\leqslant\theta\leqslant\pi/2$\,.
	Thus, expression~(\ref{eq:ConsQuant}) is
	\begin{equation*}
		\frac{\rm d}{{\rm d}x}\left(\cos^{2}\theta\,\frac{\rm d}{{\rm d}\theta}\frac{1}{\cos\theta\,V(\theta)}\right)\equiv0\,,
	\end{equation*}
	which, using the fact that
	\begin{equation*}
		\theta = \arctan(y'(x)) \implies \frac{{\rm d}\theta}{{\rm d}x} = \frac{1}{1+(y'(x))^{2}}\,y''(x)
		\quad{\rm and}\quad
		\frac{\rm d}{{\rm d}x} = \frac{{\rm d}\theta}{{\rm d}x}\frac{\rm d}{{\rm d}\theta}\,,
	\end{equation*}
	we write as
	\begin{equation*}
		\frac{y''(x)}{1+(y'(x))^{2}}\frac{\rm d}{{\rm d}\theta}\left(\cos^{2}\theta\,\frac{\rm d}{{\rm d}\theta}\frac{1}{\cos\theta\,V(\theta)}\right)\equiv0\,.
	\end{equation*}
	Hence, there are two possibilities.
	Either
	\begin{equation*}
		\underbrace{y''(x)\equiv0}_{\rm(a)} 
		\qquad{\rm or}\qquad
		\underbrace{\quad\frac{\rm d}{{\rm d}\theta}\left(\cos^{2}\theta\,\frac{\rm d}{{\rm d}\theta}\frac{1}{\cos\theta\,V(\theta)}\right)\equiv0}_{\rm(b)}\,.
	\end{equation*}
	We claim that (b) is impossible.
	Indeed, (b) holds if and only if
	\begin{equation*}
		\cos^{2}\theta\,\frac{\rm d}{{\rm d}\theta}\frac{1}{\cos\theta\,V(\theta)}\equiv a\,,
	\end{equation*}
	where $a$ is a constant.
	It follows that
	\begin{equation*}
		\frac{\rm d}{{\rm d}\theta}\frac{1}{\cos\theta\,V(\theta)}= a\sec^{2}\theta
		\iff\frac{1}{\cos\theta\,V(\theta)}= a\tan\theta + b\,,
	\end{equation*}
	where $a$ and $b$ are  constants, and
	\begin{equation*}
		\cos\theta\,V(\theta)=\frac{1}{a\tan\theta + b}
		\iff V(\theta)=\frac{1}{a\sin\theta + b\cos\theta}\,.
	\end{equation*}
	This cannot be the case, since it implies that there exists a value of $\theta^{\ast}$ for which $a\sin\theta^{\ast} + b\cos\theta^{\ast} = 0$ and $V(\theta^{\ast})\to\infty$\,.
	But $V^{3}+A(\theta)\,V+B=0$\,, so any root of this cubic equation is such that 
	\begin{equation*}
		|V|\leqslant\max\{1,|A(\theta)| + |B|\}\,.
	\end{equation*}
	Thus, it follows that (a), namely, $y''(x)\equiv0$\,, must hold, which implies that $y(x)$ is a straight line.
	For the ascent from $(0,0)$ to $(R,H)$\,, the line is $y(x) = (H/R)\,x$\,, as required.
\end{proof}

In the context of VAM, there is a corollary of Theorem~\ref{thm:Brachistochrone}.\\
\begin{corollary}
	The path of the quickest gain of altitude between two points, $(0,0)$ and $(R,H)$\,, is the straight line, $y = (H/R)\,x$\,, where $0\leqslant x\leqslant R$ and $y(R) = H$\,.
\end{corollary}

\subsection{Monotonic increase of VAM}
\label{sec:VAMmaximization}
Since the path of quickest ascent is a straight line, let us consider model~\eqref{eq:P} to determine the behaviour of VAM\,---\,under the assumption of fixed power\,---\,for slopes of interest.
To determine the behaviour, we prove the following theorem.\\
\begin{theorem}
	\label{thm:Theorem}
	For a given power, $P = P_0$\,, $V(\theta)\sin\theta$ increases monotonically for $0\leqslant\theta\leqslant\pi/2$\,.
\end{theorem}
\begin{proof}
	Expressions~\eqref{eq:PowerConstraint},~\eqref{eq:a} and~\eqref{eq:b} determine $V=V(\theta)$\,.
	Then,
	\begin{equation}
		\label{eq:d_Vsintheta_dtheta_temp1}
		\frac{\rm d}{{\rm d}\theta}\left(V(\theta)\sin\theta\right) = V'(\theta)\sin\theta + V(\theta)\cos\theta\,,
	\end{equation}
	where ${}'$ is a derivative with respect to $\theta$\,.
	The derivative of expression~\eqref{eq:PowerConstraint} is
	\begin{equation*}
		3\,V^{2}(\theta)V'(\theta) + A'(\theta)\,V(\theta) + A(\theta)V'(\theta)=0\,,
	\end{equation*}
	from which it follows that
	\begin{equation}
	\label{eq:V'(theta)}
		V'(\theta) = -\frac{A'(\theta)V(\theta)}{3\,V^{2}(\theta)+A(\theta)}\,.
	\end{equation}
	Using expression~\eqref{eq:V'(theta)} in expression~\eqref{eq:d_Vsintheta_dtheta_temp1}, and simplifying, we obtain
	\begin{equation}
		\label{eq:d_Vsintheta_dtheta_temp2}
		\frac{\rm d}{{\rm d}\theta}\left(V(\theta)\sin\theta\right) = V(\theta)\,\frac{3\,V^{2}(\theta)\cos\theta-A'(\theta)\sin\theta+A(\theta)\cos\theta}{3\,V^{2}(\theta)+A(\theta)}\,.
	\end{equation}
	Upon using expression~\eqref{eq:a}, the latter two terms in the numerator of expression~\eqref{eq:d_Vsintheta_dtheta_temp2} simplify to
	\begin{align*}
		-A'&(\theta)\sin\theta+A(\theta)\cos\theta \\
		&= \frac{mg}{\tfrac{1}{2}{\rm C_{d}A}\,\rho}\Big(-\left(\cos\theta-{\rm C_{rr}}\sin\theta\right)\sin\theta+\left(\sin\theta+{\rm C_{rr}}\cos\theta\right)\cos\theta\Big)
		\\
		&= \frac{mg\,{\rm C_{rr}}}{\tfrac{1}{2}{\rm C_{d}A}\,\rho}\,.
		\end{align*}
	Hence,
	\begin{equation}
		\label{eq:d_Vsintheta_dtheta}
		\frac{\rm d}{{\rm d}\theta}\left(V(\theta)\sin\theta\right) 
		= 
		V(\theta)\,\frac{3\,V^{2}(\theta)\cos\theta+\dfrac{mg\,{\rm C_{rr}}}{\tfrac{1}{2}{\rm C_{d}A}\,\rho}}{3\,V^{2}(\theta)+A(\theta)}>0\,,
		\quad{\rm for}\quad
		0\leqslant\theta\leqslant\frac{\pi}{2}\,,
	\end{equation}
	as required.
	Since expression~\eqref{eq:d_Vsintheta_dtheta} is positive, VAM is\,---\,within the slope range of interest\,---\,a monotonically increasing function of $\theta$\,.
\end{proof}
From this theorem, it follows that the steeper the incline that a rider can climb with a given power, the greater the value of the corresponding VAM, as illustrated in Figure~\ref{fig:FigThm1and2}.
\subsection{VAM as a function of cadence and gear ratio}
In view of Theorem~\ref{thm:Theorem}, we conclude that, for a fixed power, there is no specific value of $\theta$ that maximizes expression~\eqref{eq:VAM} and hence, there is no local maximum for VAM. 
In other words\,---\,based on expression~\eqref{eq:P}, alone\,---\,we cannot specify the optimal slope to maximize VAM.
However, this is not the case if we include other considerations, such as
\begin{equation}
	\label{eq:fv}
	P = f\,v\,,
\end{equation}
where $f$ is the force applied to the pedals and $v$ is the circumferential speed of the pedals.
These are measurable quantities; they need not be modelled in terms of the surrounding conditions, such as the mass of the rider, strength and direction of the wind or the steepness of an uphill.
Herein, the circumferential speed is
\begin{equation}
	\label{eq:v}
	v=2\pi\,\ell_c\,c\,,
\end{equation}
where $\ell_c$ is the crank length, $g_r$ is the gear ratio, $r_w$ is the wheel radius and $c$ is the cadence, and the corresponding ground speed is
\begin{equation}
	\label{eq:V}
	V=2\pi\,r_{w}\,c\,g_r\,;
\end{equation}
these expressions allow us to examine VAM as a function of the cadence and gear ratio.

To examine VAM as a function of cadence and gear ratio, we require typical values for the parameters that constitute the power model.
Thus, we consider~\citet[Section~3]{DanekEtAl2021}, which details a numerical optimization process for the estimation of $\rm C_dA$\,, $\rm C_{rr}$ and $\lambda$ using model~\eqref{eq:PV} along an 8-kilometre climb of nearly constant inclination in Northwest Italy.
The pertinent values are as follows\footnote{%
	For the convenience of the reader, all numerical values resulting from operations on model parameters are rounded to four significant figures.
}.
For the bicycle-cyclist system, $m=111\,{\rm kg}$, $\rm C_{d}A=0.2702\,{\rm m}{}^2$, $\rm C_{rr}=0.01298$\,, $\lambda=0.02979$\,.
For the bicycle, $\ell_c = 0.175\,{\rm m}$, $g_{r}=1.5$ and $r_w = 0.335\,{\rm m}$.
For the external conditions, $g=9.81\,{\rm m/s^2}$ and $\rho=1.168\,{\rm kg/m^3}$, which corresponds to an altitude of 400\,m.

Let us suppose that the power that the cyclist maintains during a climb a power of $P_0=375\,{\rm W}$ and that the cadence is $c = 1$\,, which is 60\,rpm, and\,---\,in accordance with expression~\eqref{eq:v}\,---\, results in $v = 1.100\,{\rm m/s}$.
The force that the rider must apply to the pedals is $f = 341.0\,{\rm N}$.
In accordance with expression~\eqref{eq:V}, the ground speed is $V=3.157\,{\rm m/s}$ and, inserting $V$ into expression~\eqref{eq:PowerConstraint}, we obtain
\begin{equation*}
	46.00\cos\theta + 3544\sin\theta - 369.9 = 0\,,
\end{equation*}
whose solution is $\theta = 0.09160$\,rad\,$=5.248^\circ$\,, which results in ${\rm VAM} = 1040$\,m/h; herein, $\theta$ corresponds to a grade of 9.184\,\%\,. 

Now, let us consider the effect on this value of VAM due to varying the cadence and gear ratio.
To do so, we use expression~(\ref{eq:V}) in expressions~\eqref{eq:VAM} and \eqref{eq:PowerConstraint} to obtain
\begin{equation}
\label{eq:VAM(c,g_r,theta)}
	{\rm VAM}=7578\,c\,g_r\sin\theta
\end{equation}
and 
\begin{equation}
	\label{eq:F(c,g_r,theta)} 
	1.517\,(c\,g_r)^3 + 30.66\,c\,g_{r}\cos\theta + 2362\,c\,g_{r}\sin\theta - 375.0 = 0\,,
\end{equation}
respectively.
In both expressions~(\ref{eq:VAM(c,g_r,theta)}) and (\ref{eq:F(c,g_r,theta)}), $c$ and $g_r$ appear  as a product.

For each $c\,g_r$ product, we solve equation~\eqref{eq:F(c,g_r,theta)} to obtain the corresponding value of~$\theta$\,, which we use in expression~\eqref{eq:VAM(c,g_r,theta)} to calculate the VAM.
We consider $c\,g_r\in[0.5\,,3.75]$\,; the lower limit represents $c=0.5$\,, which is 30\,rpm, and $g_r = 1$\,; the upper limit represents $c=1.5$\,, which is 90\,rpm, and $g_r = 2.5$\,.
Following expression~(\ref{eq:v}), the corresponding circumferential speeds are $v=0.5498$\,m/s and $v=1.649$\,m/s\,.
Since $P_0=375$\,W\,, in accordance with expression~(\ref{eq:fv}), the forces applied to the pedals are $f=682.1$\,N and $f=227.4$\,N\,, respectively.

\begin{figure}[h]
\centering
\includegraphics[scale=0.4]{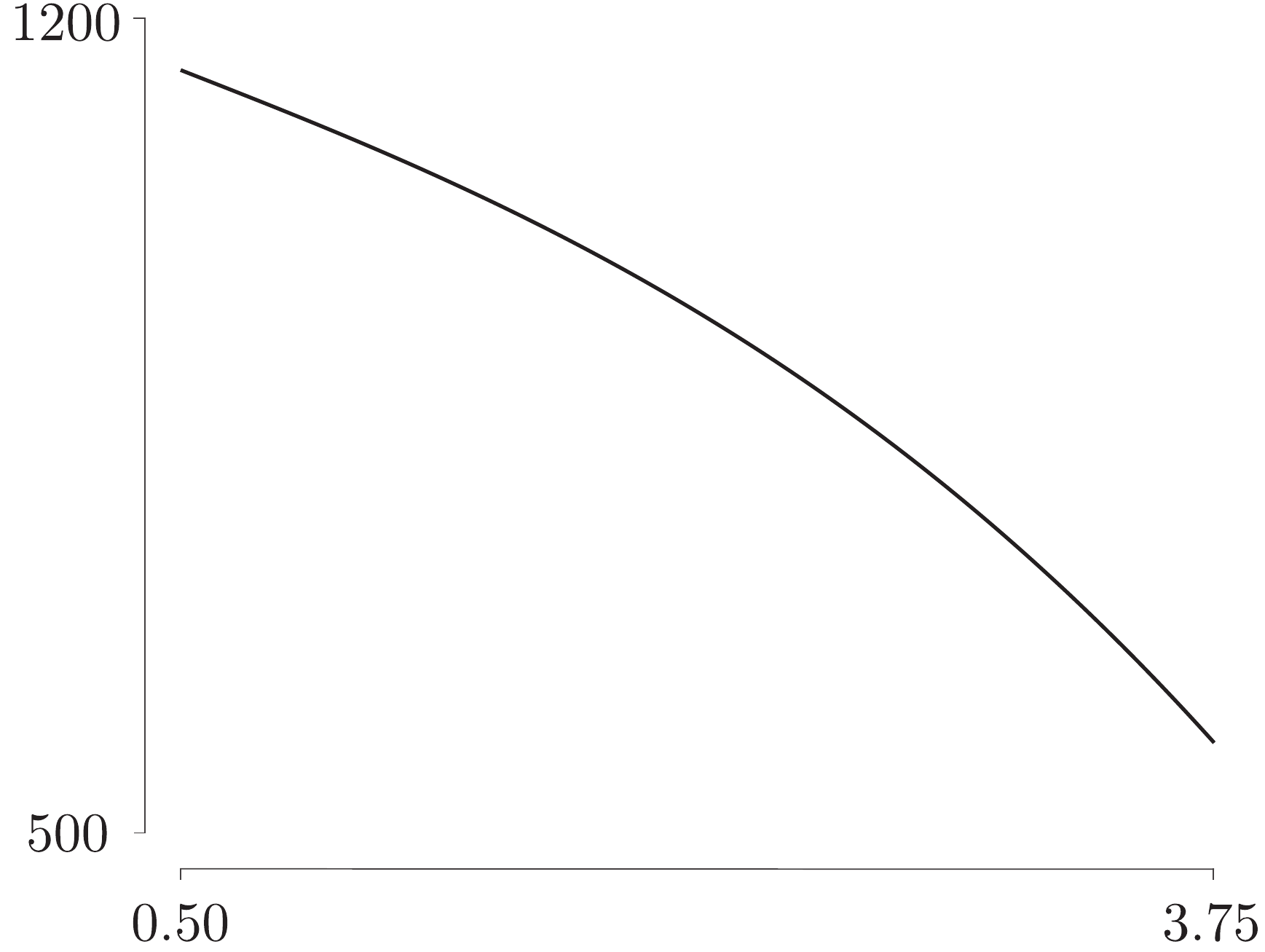}
\caption{\small%
VAM as function of $c\,g_r$\,, with $P_0=375$\,W}
\label{fig:VAM}
\end{figure}

Examining Figure~\ref{fig:VAM}, we conclude that lowering $c\,g_r$ increases the slope of the  climbable incline and, hence\,---\,for a given value of power\,---\,increases the VAM, as expected in view of Theorem~\ref{thm:Theorem}.

The VAM is determined not only by the power sustainable during a climb, but also the steepness of that climb.
As stated in Theorem~\ref{thm:Theorem}, the steeper the slope that can be climbed with a given power, the greater the VAM.
 Hence, as illustrated in Figure~\ref{fig:VAM}, the highest value corresponds to the lowest cadence-gear product.
 A maximization of VAM requires an optimization of the force applied to the pedals and their circumferential speed, along a slope, in order to maintain a high power.

\section{Discussion}

\begin{wrapfigure}{2}{0.525\textwidth}
\centering
\includegraphics[width=0.475\textwidth]{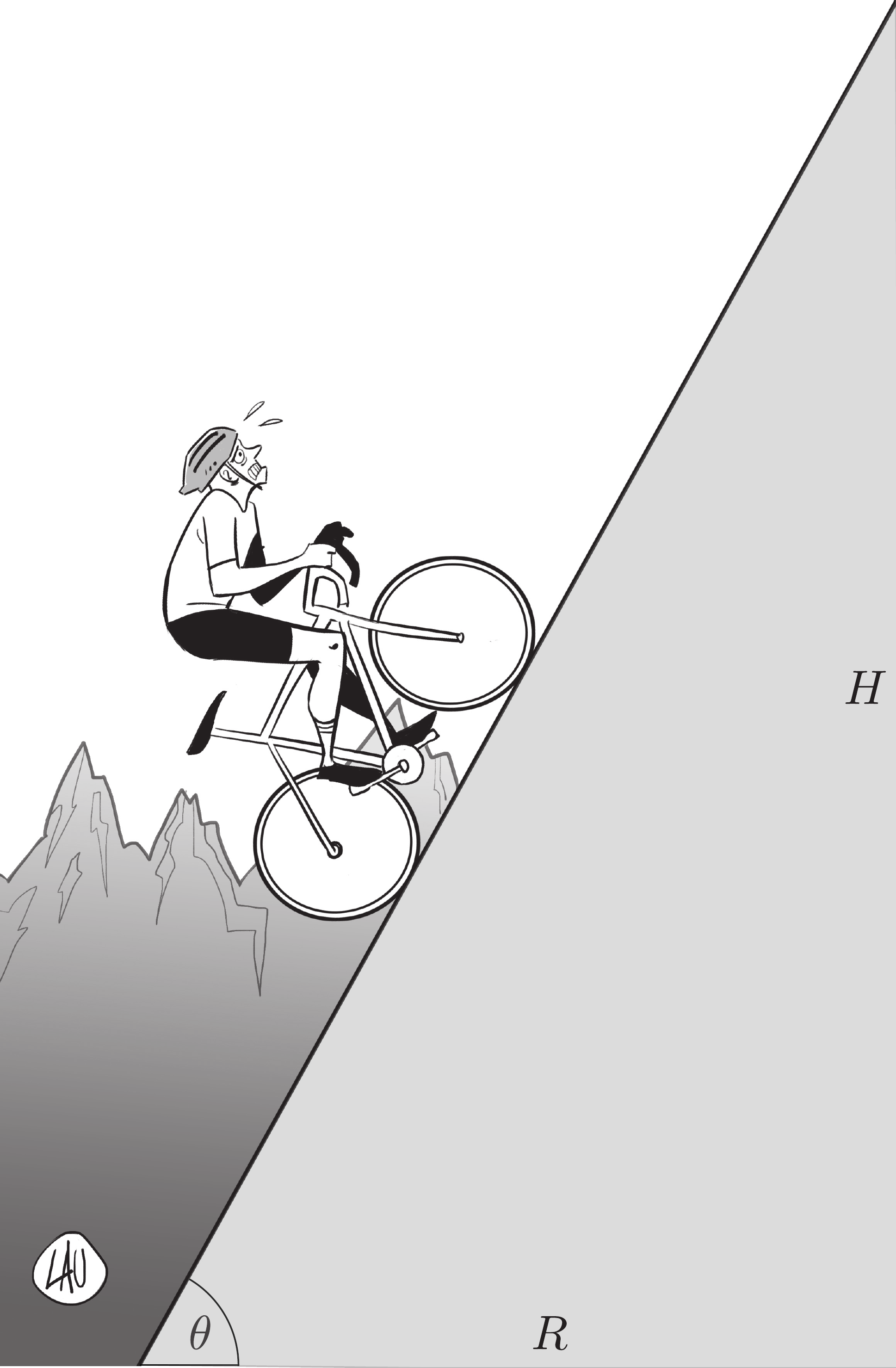}
\caption{\small%
	Illustration of Theorems~\ref{thm:Brachistochrone} and \ref{thm:Theorem}}
\label{fig:FigThm1and2}
\end{wrapfigure}

Let us illustrate Theorems~\ref{thm:Brachistochrone} and~\ref{thm:Theorem} in the context of climbs chosen for Everesting, which consists of riding an uphill repeatedly in a single activity until reaching the total ascent of 8,848~metres.

In June of 2020, Lachlan Morton \citep{MortonRecord} established an Everesting record.
To do so\,---\,in a manner  consistent with Theorem~\ref{thm:Theorem}\,---\,he chose a steep slope, whose average gradient is~$11.1\%$\,.
This record was established\,---\,in a manner consistent with Theorem~\ref{thm:Brachistochrone} and its corollary\,---\,on a hill whose average and maximum gradients are $11.1\%$ and $13.2\%$\,, respectively.
A more recent record\,---\,by Sean Gardner~\citep{GardnerRecord} in October of 2020\,---\,was established on a hill whose average and maximum gradients are $15.5\%$ and $22.6\%$\,, respectively.
According to Theorem~\ref{thm:Brachistochrone} and its corollary, which are illustrated in Figure~\ref{fig:FigThm1and2}, a more steady uphill would have been preferable.
Let us emphasize that Theorem~\ref{thm:Brachistochrone} and its corollary are statements of mathematical physics, whose implications might be adjusted for an optimal strategy.
For instance, it might be preferable\,---\,from a physiological viewpoint\,---\,to vary the slope to allow periods of respite.

Also, let us revisit expression~\eqref{eq:VAM(c,g_r,theta)}, whose general form, in terms of $c\,g_r$\,, is
\begin{equation*}
8\,\pi^{3}r_{w}^{3}(c\,g_r)^{3} + \frac{2\,\pi\,r_{w}\,m\,g\,(\sin\theta+{\rm C_{rr}}\,\cos\theta)}{\tfrac{1}{2}{\rm C_{d}A}\,\rho}\,c\,g_{r} - \frac{P_{0}(1-\lambda)}{\tfrac{1}{2}{\rm C_{d}A}\,\rho} = 0
\,.
\end{equation*}
Using common values for most quantities, namely, \mbox{${\rm C_{d}A} = 0.27\,{\rm m}^2$}, \mbox{${\rm C_{rr}} = 0.02$}, \mbox{$\lambda = 0.03$}, \mbox{$\rho = 1.2\,{\rm kg/m}^3$}, \mbox{$r_{w} = 0.335\,{\rm m}$}, and, for convenience of a concise expression, below, we invoke the small-angle approximation, which results\,---\,in radians\,---\,in $\sin\theta\approx\theta$ and $\cos\theta\approx1-\theta^{2}/2$\,, to obtain
\begin{equation*}
\theta = 50-15.5610\sqrt{10.3326+\frac{0.0302\,c\,g_r-0.0194\,P_0}{m\,c\,g_r}}\,,
\end{equation*}
where $m$ is the mass of the bicycle-cyclist system, and the product, $c\,g_r$\,, might be chosen to correspond to the maximum sustainable power,~$P_0$\,.
This product is related to power by expression~(\ref{eq:fv}), since $c$ is proportional to $v$\,, and $g_r$ to $f$\,.
For any value of $c\,g_r$\,, there is a unique value of $P_0$\,, which depends only on the power of a rider.

As expected, and as illustrated in Figure~\ref{fig:FigTheta}, the slope decreases with~$m$\,, {\it ceteris paribus}.
Also, as expected, and as illustrated in Figure~\ref{fig:FigVAM}, so does VAM.
For both figures, the abscissa, expressed in kilograms, can be viewed as the power-to-weight ratio from $6.25$ to $3.41$\,, whose limits represent the values sustainable\,---\,for about an hour\,---\,by an elite and a moderate rider, respectively.

\begin{figure}[h]
	\centering
	\begin{subfigure}[b]{0.49\textwidth}
		\centering
		\includegraphics[scale=0.4]{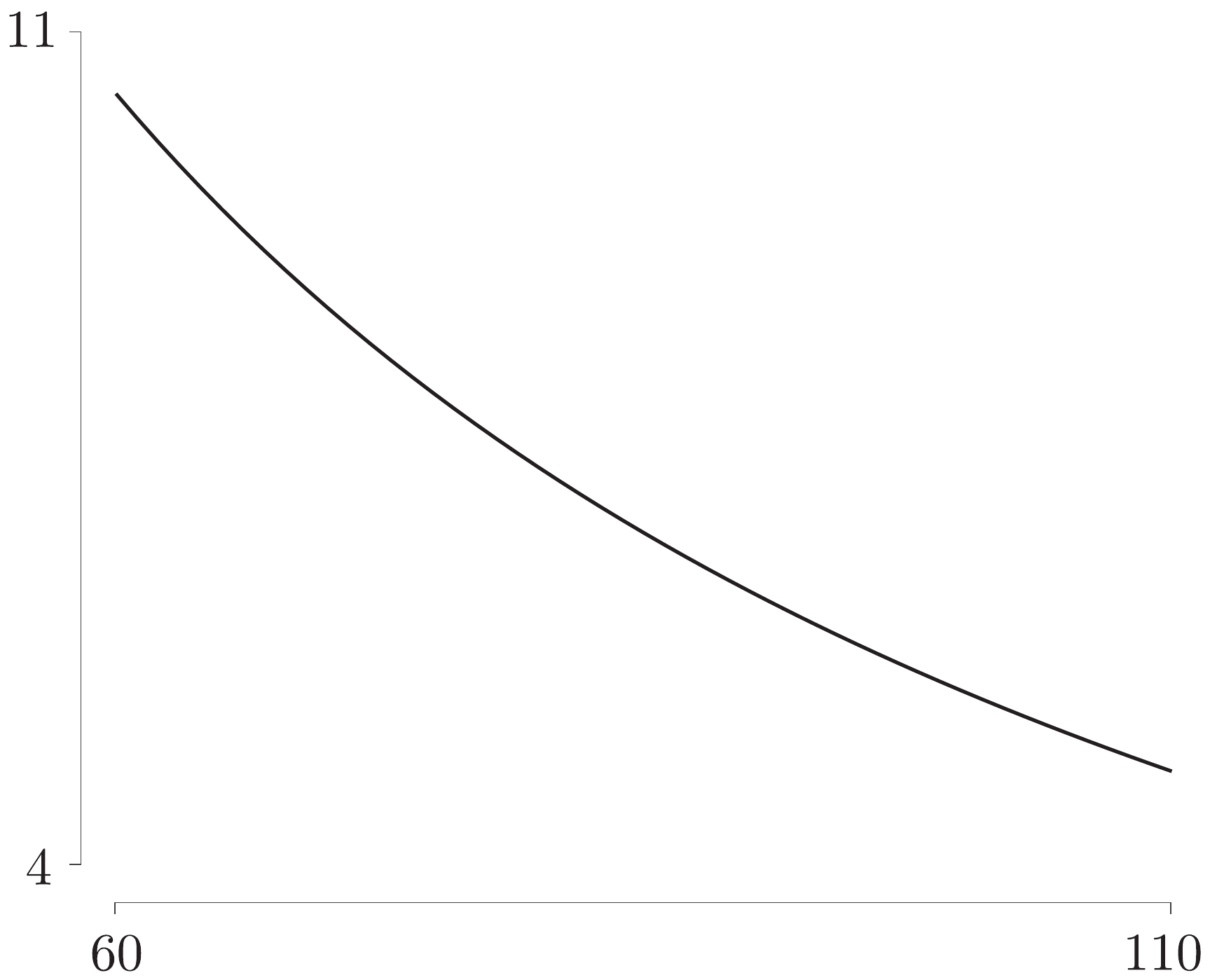}
		\caption{\small Slope}
		\label{fig:FigTheta}
	\end{subfigure}
	\hfill
	\begin{subfigure}[b]{0.49\textwidth}
		\centering
		\includegraphics[scale=0.4]{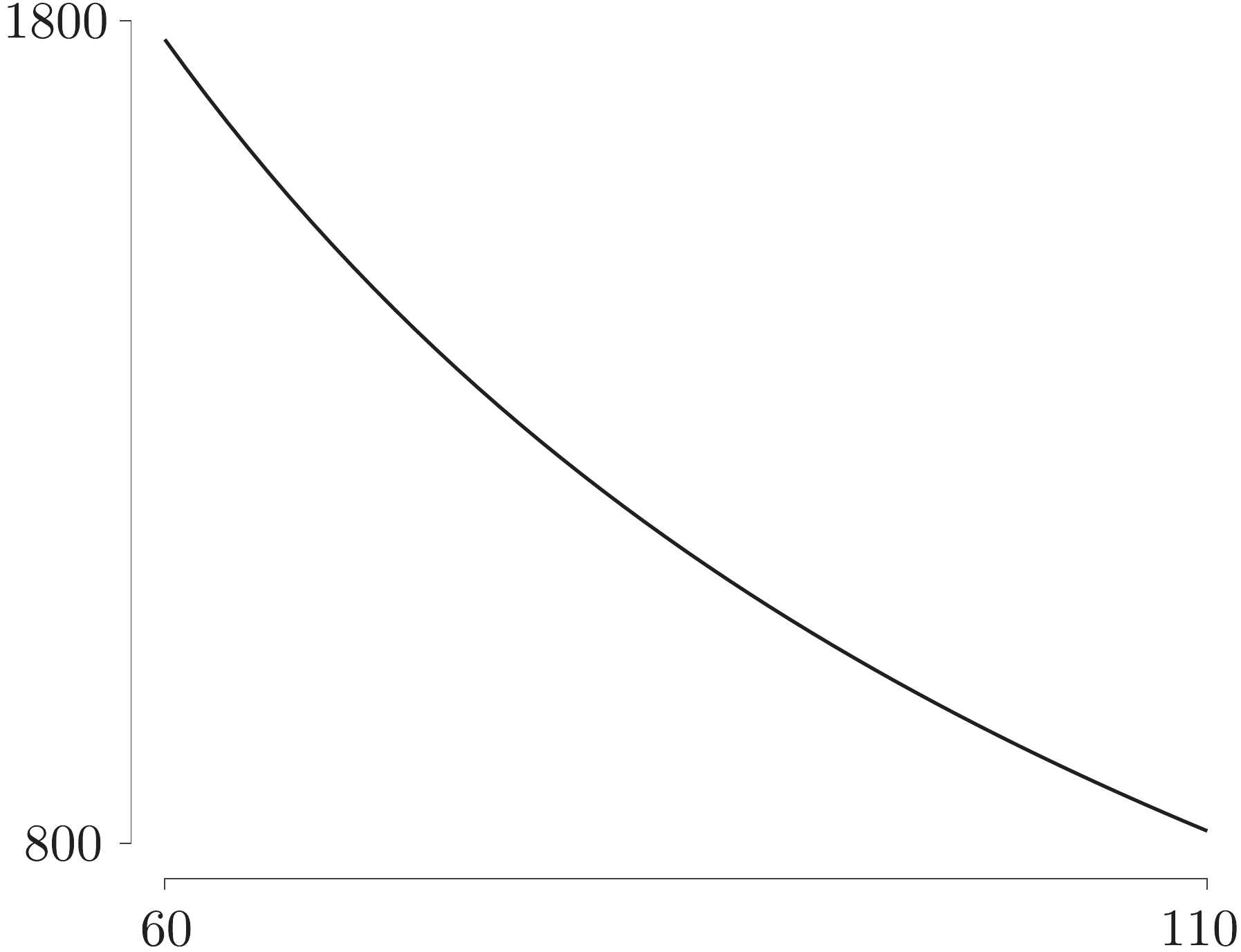}
		\caption{\small VAM}
		\label{fig:FigVAM}
	\end{subfigure}
	\caption{\small
		Slope and VAM as functions of $m$\,, with $P_0=375$\,W and $c\,g_r=2.25$}
	\label{fig:CronoMaps}
\end{figure}

\section{Conclusions}
The presented results constitute a mathematical-physics background upon which various strategies for the VAM maximization can be examined.
Following fundamental properties of the calculus of variations and the rigor of mathematical proof, we prove that\,---\,for a fixed power\,---\,the quickest ascent between two points is a straight line and that VAM increases monotonically with slope.
The latter conclusion yields that there is no specific slope value that maximizes VAM.
Rather, we demonstrate that its maximization is determined by the maximum power sustainable by the cyclist during a climb.
It follows that the lower the product of cadence and gear ratio, $c\,g_r$, the greater the slope of the climbable incline for a fixed power and, in turn, the greater the VAM.

In the context of physiological considerations, one could examine the maximum sustainable power as a function of both the force applied to the pedals and their circumferential speed.
For instance, depending on the physical ability of the cyclist, one might choose a less steep slope to allow a higher $c\,g_r$ product, whose value allows to generate and sustain a higher power.
One might also consider a very short and an exceptionally steep slope, while keeping in mind the issue of measurement errors due to the shortness of the segment itself.
In all cases, however, the slope needs to be constant.

\section*{Acknowledgements}
We wish to acknowledge Elena Patarini, for her graphic support, Roberto Lauciello, for his artistic contribution, and Favero Electronics for inspiring this study by their technological advances and for supporting this work by providing us with their latest model of Assioma Duo power meters.
\bibliographystyle{apa}
\bibliography{BSSbici6.bib}
\end{document}